\documentclass[a4paper,USenglish,cleveref,autoref,thm-restate,pdfa]{lipics-v2021}
\hideLIPIcs\def\subjclassHeading{}\ccsdesc{}\global\renewcommand\ccsdesc[2][100]{}
\nolinenumbers 
\title{First-Fit Coloring of Forests in Random Arrival Model}
\allowdisplaybreaks

\author{Bart{\l}omiej Bosek}
{Institute of Theoretical Computer Science, Faculty of Mathematics and Computer Science, Jagiellonian University, Krak{\'o}w, Poland}
{bartlomiej.bosek@uj.edu.pl}
{https://orcid.org/0000-0001-8756-3663}
{Partially supported by grant no.~2023/51/B/ST6/02833 from National Science Centre, Poland.}

\author{Grzegorz Gutowski}
{Institute of Theoretical Computer Science, Faculty of Mathematics and Computer Science, Jagiellonian University, Krak{\'o}w, Poland}
{grzegorz.gutowski@uj.edu.pl}
{https://orcid.org/0000-0003-3313-1237}
{Partially supported by grant no.~2023/51/B/ST6/02833 from National Science Centre, Poland.}

\author{Micha{\l} Laso\'n}
{Institute of Mathematics of the Polish Academy of Sciences, ul. \'Sniadeckich 8, 00-656 Warszawa, Poland}
{michalason@gmail.com}
{https://orcid.org/0000-0003-4830-2270}
{Partially supported by grant no.~2019/34/E/ST1/00087 from National Science Centre, Poland.}

\author{Jakub Przyby{\l}o}
{AGH University of Krakow, Faculty of Applied Mathematics, al.\ A.\ Mickiewicza 30, 30-059 Krak{\'o}w, Poland}
{jakubprz@agh.edu.pl}
{https://orcid.org/0000-0002-1262-7017}
{}

\authorrunning{B. Bosek, G. Gutowski, M. Laso\'n, and J. Przyby{\l}o}
\Copyright{Bart{\l}omiej Bosek, Grzegorz Gutowski, Micha{\l} Laso\'n, and Jakub Przyby{\l}o}


\usepackage{tikz}
\usetikzlibrary{arrows,math,patterns,shapes,fit,decorations.pathreplacing}
\usepackage{todonotes}
\usepackage{nicefrac}

\let\leq\leqslant
\let\geq\geqslant
\let\le\leqslant
\let\ge\geqslant
\let\epsi\varepsilon
\let\rho\varrho

\newcommand{\brac}[1]{{\left(#1\right)}}

\newcommand{\set}[1]{\left\{#1\right\}}

\newcommand{\norm}[1]{{\left|#1\right|}}

\newcommand{\ceil}[1]{{\left\lceil #1 \right\rceil}}

\newcommand{\Oh}[1]{O\brac{#1}}
\newcommand{\oh}[1]{o\brac{#1}}
\newcommand{\Om}[1]{\Omega\brac{#1}}

\newcommand{\OT}[1]{\Theta\brac{#1}}

\DeclareMathOperator{\poly}{poly}

\DeclareMathOperator{\EXP}{\mathbb{E}}

\DeclareMathOperator{\RFF}{\mathrm{RFF}}
\DeclareMathOperator{\AFF}{\mathrm{FF}}

\begin{document}

\maketitle

\begin{abstract}
We consider a graph coloring algorithm that processes vertices in order taken uniformly at random and assigns colors to them using First-Fit strategy.
We show that this algorithm uses, in expectation, at most
$(1+\oh{1})\cdot \ln n \, / \ln\ln n$
different colors to color any forest with $n$ vertices.
We also construct a family of forests that shows that this bound is best possible.
\end{abstract}

\section{Introduction}

A \emph{proper $k$-coloring} of a graph $G=(V,E)$ is a function $c : V \mapsto \{1,2,\ldots, k\}$ that assigns a color $c(v)$ to each vertex $v \in V$ so that any adjacent vertices are colored differently, i.e., for each edge $\set{u,w} \in E$, $c(u) \neq c(w)$ is satisfied.
For a given graph $G$, the smallest number $k$ for which $G$ admits a proper $k$-coloring is the \emph{chromatic number} of
$G$ and is denoted by $\chi(G)$.
Graph coloring is one of the most prominent disciplines within graph theory, with plenty of variants, applications, and deep connections to theoretical computer science.
Coloring problems arise naturally in various job scheduling and resource allocation optimization scenarios.

The graph coloring problem is also very popular and well motivated in the online setting, with applications in job scheduling, dynamic storage allocation and resource management~\cite{Kierstead98,Marx04,Lata02}.
In the \emph{online graph coloring} problem, an online algorithm receives as input a graph $G = (V,E)$ presented in an online fashion.
The vertices of $V$ are revealed one after one, in a \emph{presentation order} $v_1 \ll v_2 \ll \ldots \ll v_{n}$.
When a new vertex $v_t$ is revealed in the $t$-th round, for $1 \leq t \leq n$, all the edges connecting $v_t$ with vertices in $V_{t-1}=\{ v_1, \ldots, v_{t-1} \}$ are also revealed.
An online algorithm $\mathcal{A}$ has to immediately assign a feasible color to $v_{t}$, that is, a color that is different from those assigned to the neighbors of $v_{t}$ in $V_{t-1}$.
The goal is to minimize the total number of colors used.

The performance of an online graph coloring algorithm $\mathcal{A}$ is often measured by its \emph{competitive ratio}.
If we denote by $\chi_{\mathcal{A}}(G, \ll)$ the number of colors used by $\mathcal{A}$ when coloring a graph $G$ in a presentation order $\ll$, then the competitive ratio of $\mathcal{A}$ is the maximum ratio $\chi_{\mathcal{A}}(G, \ll)\,/\,\chi(G)$ over all graphs $G$ and all orders $\ll$.
This means, that for a single graph, the analysis focuses on the worst case scenario, and the presentation order is often considered to be selected by an adversary.
If $\mathcal{A}$ is a randomized algorithm, then it is usual to measure the competitive ratio with respect to the expected number of colors used when coloring $G$ in order $\ll$ over all the random choices of $\mathcal{A}$.

Contrary to the offline setting, even coloring $2$-colorable graphs is not easy for online algorithms.
It is impossible to construct an online algorithm that would use any constant number of colors to color all $2$-colorable graphs.
Thus, the research focuses on restricted graph classes, and competitive ratio is often expressed as a function of the number of vertices in a graph.
In fact, even for inputs restricted to online forests with $n$ vertices, the best competitive ratio that can be achieved by either randomized or deterministic online coloring algorithm is $\OT{\log n}$ colors~\cite{Bean76,GyarfasL88,AlbersS17}.
The optimal number of colors used by any deterministic online algorithm to color $2$-colorable graphs with $n$ vertices is known to be somewhere between $2\log_2n-10$ (by a result of Gutowski~et~al.~\cite{GutowskiKMZ14}) and $2\log_2n$ (see a paper by Lov\'asz, Saks, and Trotter~\cite{LovaszST89}).

An algorithm known as First-Fit is arguably the simplest and the most understood of all online deterministic coloring algorithms.
When a vertex $v_t$ is revealed in the $t$-th round, First-Fit picks the least positive integer $i$ which does not occur as a color of any of the previously colored neighbors of $v_t$ and assigns $i$ as a color of $v_t$.
First-Fit performs well on online trees, achieving the competitive ratio of $\OT{\log n}$ within this class.
To be more precise, Bean~\cite{Bean76} and, independently, Gy\'arf\'as and Lehel~\cite{GyarfasL88} proved that First-Fit uses at most $\log_2 n + 1$ colors on forests with $n$ vertices.
Irani~\cite{Irani94} generalized this result, showing that First-Fit uses $O(d \log n)$ colors on $d$-degenerate graphs with $n$ vertices.
Later, Balogh et.\;al.~\cite{BaloghHLY08} with corrections of Chang and Hsu~\cite{ChangH12} improved the above result to at most $(\log_{\frac{d+1}{d}} n+2)$ colors.
For bipartite graphs, there is an easy construction by Lov\'asz, Saks, and Trotter~\cite{LovaszST89} that shows that First-Fit 
uses as many as 
$\frac{n}{2}$ colors to color a $2$-colorable graph with $n$ vertices.

In this paper, we are interested in the \emph{random arrival model} for online algorithms, that tries to focus on the average case, rather than the worst case scenario.
In this model, the presentation order of a graph is not determined by an adversary, but instead it is selected uniformly at random from all possible permutations of the vertex set.
The \emph{performance ratio} of a deterministic online algorithm $\mathcal{A}$ on a graph $G$ is measured by the expectation $\EXP_{\ll}\left[\frac{\chi_{\mathcal{A}}(G, \ll)}{\chi(G)}\right]$ over the random choice of $\ll$.
The performance ratio of $\mathcal{A}$ on a class of graphs $\mathcal{G}$ is the the maximum performance ratio taken over all graphs in $\mathcal{G}$.

In particular, our intention is to start the systematic study of the performance ratio of First-Fit coloring algorithm in random arrival model.
For a class of graphs $\mathcal{G}$, let $\RFF_{\mathcal{G}}(n)$ denote the maximum performance ratio of First-Fit taken over all graphs in $\mathcal{G}$ with $n$ vertices, i.e.,
\[
\RFF_{\mathcal{G}}(n) = \max_{G \in \mathcal{G}, \norm{G}=n} \EXP_{\ll}\left[\frac{\chi_{\mathrm{FF}}(G,\ll)}{\chi(G)}\right]\text{.}
\]
For an easy comparison to the adversarial model, we define:
\[
\AFF_{\mathcal{G}}(n) = \max_{G \in \mathcal{G}, \norm{G}=n} \max_{\ll}\left[\frac{\chi_{\mathrm{FF}}(G,\ll)}{\chi(G)}\right]\text{.}
\]
In this paper, we prove the following result, which establishes the performance of First-Fit on the class of forests.
\begin{theorem}
\label{thm:main}
For the class $\mathcal{F}$ of forests, we have
$
\RFF_{\mathcal{F}}(n) = (\nicefrac{1}{2}\pm\oh{1})\cdot \ln n \,/\, \ln\ln n \text{.}
$
\end{theorem}

As mentioned above, we have $\AFF_\mathcal{F}(n) = \OT{\log n}$.
\cref{thm:main} shows that the randomization of the presentation order gives a noticeable, yet rather moderate increase in performance compared to the adversarial model.

One can also consider the following natural off-line graph coloring algorithm.
Given a graph $G$, it selects an order $\ll$ of vertices of $G$ uniformly at random.
Then it uses First-Fit strategy to color vertices of $G$ in order $\ll$.
\cref{thm:main} shows that this algorithm uses $\Oh{\log n\,/\,\log\log n}$ colors in expectation to color any forest.
This algorithm does not compete with easy $2$-coloring algorithms based on graph traversal.
Our objective, however, is to establish the groundwork for a systematic analysis of First-Fit, and other simple online algorithms, in the random arrival model for other graph classes, such as $d$-degenerate graphs, bipartite graphs, and subsequently, $k$-colorable graphs.
It should be noted that for $3$-colorable graphs, the most effective known randomized online algorithm utilizes expected $\Oh{n^{\nicefrac{1}{2}}}$ colors, as proved by Halld\'{o}rsson~\cite{Halldorsson97}.
Considering that the best lower bound for $3$-colorable graphs is in the order of $\Om{\log^2n}$, as demonstrated by Vishwanathan~\cite{Vishwanathan92}, it appears that the direction we are pursuing holds promise for intriguing outcomes.

Another rationale for this line of inquiry is that First-Fit in random arrival model serves as an illustration for the following distributed coloring algorithm that works in the synchronous model.
Each vertex of the graph represents a computation node, and edges represent communication links.
To divide nodes into independent subsets, each node performs a sleep for a random number of units of time, and then assigns to itself the first possible number not assigned to any of the neighbors.
This algorithm uses a small number of messages and is quite fast.
Although we refrain from further delving into this setting, we underscore the versatility afforded by our analysis.

The paper is organized as follows.
In \cref{sec:upper-bound} we show that First-Fit uses at most $(1+\oh{1})\cdot \ln n \, / \, \ln\ln n$ different colors in expectation on any forest with $n$ vertices.
In \cref{sec:lower-bound} we construct a family of trees for which First-Fit uses $(1-\oh{1})\cdot \ln n \, / \, \ln\ln n$ different colors in expectation.
The final section contains brief comments and some open problems regarding the analyzed problem.

\section{Upper bound}
\label{sec:upper-bound}

In this section we show that First-Fit uses at most $\Oh{\log n \, / \, \log\log n}$ different colors in expectation on any forest with $n$ vertices.
For this purpose let us first make a basic observation.
For any fixed graph $G=(V,E)$ and any order $\ll$ of $V$, we denote by $G^\ll$ the acyclic directed graph obtained from $G$ by orienting
every edge $\{u,v\} \in E$ so that it is oriented from $u$ to $v$ if and only if $u \ll v$.

\begin{observation}\label{obs:onepath}
For any graph $G=(V,E)$, a presentation order~$\ll$ of $V$, a positive integer~$i$, and a vertex $v\in V$, if First-Fit assigns color $i$ to $v$ when coloring $G$ in order $\ll$ then there is a directed path in $G^\ll$ with $i$ vertices and $v$ as the last vertex.
\end{observation}

\begin{proof}
    We prove the observation by induction with respect to $i$.
    The statement is trivial for $i=1$, so let us assume $i\geq 1$ and that the thesis holds for all values up to $i$.
    Suppose that First-Fit assigns color $i+1$ to vertex $v$ in some round of the algorithm.
    At this point, $v$ has a neighbor $u$ with $u \ll v$ that was assigned with color $i$ in some earlier round.
    By induction hypothesis, there is a directed path $P$ in $G^\ll$ with $i$ vertices and $u$ as the last vertex.
    Since $u\ll v$, and $G^\ll$ is acyclic we know that $v$ does not belong to $P$.
    Thus, we get the desired path by extending $P$ with directed edge $(u,v)$.
\end{proof}

We call a simple path $P$ in a directed graph to be \emph{bidirected} if it is either a directed path, 
or $P$ can be split into two directed paths, each starting at one of the end points of $P$, and both sharing the same last vertex.

\begin{corollary}\label{cor:twopaths}

For any forest $T=(V,E)$, a presentation order~$\ll$ of $V$, and a positive integer $i \ge 2$, if  First-Fit uses color $i$ when coloring $T$ in order $\ll$ then there is a bidirected path in $T^\ll$ with $2i-2$ vertices.
\end{corollary}
\begin{proof}
    For $i=2$, by \cref{obs:onepath} we get a directed path with $2$ vertices.
    Suppose $i \ge 3$, and that First-Fit assigns color $i$ to some vertex $v$ of $T$.
    At this point, $v$ has a neighbor $u$ that is colored $i-1$, and a different neighbor $w$ that is colored $i-2$.
    By \cref{obs:onepath} there is a directed path with $i-1$ vertices and $u$ as the last vertex, and a directed path with $i-2$ vertices and $w$ as the last vertex.
    As $T$ is a forest, these paths are vertex disjoint.
    Since $u\ll v$, and $w\ll v$, we get the desired bidirected path in $T^\ll$ by extending both paths with directed edges $(u,v)$ and $(w,v)$.
\end{proof}

\begin{lemma}
    \label{lem:upper-bound}
    For the class $\mathcal{F}$ of forests, every $n\ge 3$, and an $\alpha_n=\frac{\ln\ln\ln n + 1}{\ln\ln n - \ln\ln\ln n - 1}$, 
    we have:
    \[\RFF_{\mathcal{F}}(n) \le \frac{(1+\alpha_n) \ln n}{2\ln\ln n}+\frac{3}{2}\text{.}
    \]
\end{lemma}
\begin{proof}
It is straightforward to verify that $\RFF_{\mathcal{F}}(3) =1$ and $\RFF_{\mathcal{F}}(4) < 2$, hence the lemma holds for $n=3,4$. Let us assume that $n\geq 5$. Then, since the function $\ln \ln x-\ln\ln\ln x-1 > 0$ for every $x>e$ except $x=e^e$, then $\ln \ln n-\ln\ln\ln n-1 > 0$, and hence $\alpha_n>0$.
Consider any forest $T$ with $n$ vertices.
Let $k=\ceil{\frac{(1+\alpha_n)\ln n}{\ln\ln n}}$, and observe that our goal is to prove that First-Fit uses at most $k+2$ colors in expectation when coloring $T$ in a random order.
Indeed, we can assume that $\chi(T) = 2$, as otherwise $T$ is an independent set, and First-Fit uses only one color when coloring $T$ in any order.
Note that $(1+\alpha_n)\ln n / \ln\ln n > \ln n / \ln\ln n > 1$. Hence, $k\geq 2$.
By \cref{cor:twopaths}, for every order $\ll$ of the vertices of $T$ and a positive integer $i \ge 2$, if $\chi_{FF}(T,\ll)=i$, then there is a bidirected path with $2i-2$ vertices in $T^\ll$.
Consider any two vertices $x$, $y$ of $T$.
There is at most one simple path in $T$ with end points $x$ and $y$.
If this path exists, and is a path with exactly $2i-2$ vertices, then the probability that this path is bidirected in $T^\ll$ for a random order $\ll$ is exactly $2^{2i-3}\,/\,(2i-2)!$.
Otherwise, the probability that there is such a path with end points $x$ and $y$ equals $0$.

Thus, by the union bound, the probability that there is at least one bidirected path with $2i-2$ vertices in $T^\ll$
is upper bounded by $\brac{n^2\,/\,2}\cdot\brac{2^{2i-3}\,/\,(2i-2)!}$.
Hence, the expected number of colors used by First-Fit in a random order satisfies:
\begin{align}
    \EXP_{\ll}\left[\chi_{FF}(T,\ll)\right] & =
    \sum_{i=1}^{\infty} i \cdot \mathbb{P}(\chi_{FF}(T,\ll)=i) \le
    k + 1 + \sum_{i=k+2}^\infty \frac{i \cdot n^2 \cdot 2^{2i-4}}{(2i-2)!} \le \nonumber\\
    & = k + 1 + \frac{n^2\cdot 4^{k-1}}{(2k)!} \cdot \sum_{i=k+2}^{\infty} \frac{i\cdot 2^{2(i-(k+1))}}{(2k+1)\cdot(2k+2)\cdot\ldots\cdot(2i-3)\cdot(2i-2)} \le \nonumber\\
    & \le k + 1 + \frac{n^2\cdot 4^{k-1}}{(2k)!} \cdot \sum_{i=k+2}^{\infty} \frac{2^{2(i-(k+1))}}{1\cdot(2k+2)\cdot\ldots\cdot(2i-3)\cdot 2} \le \nonumber\\ & \le k + 1 + \frac{n^2\cdot 4^{k-1}}{(2k)!} \cdot \sum_{i=1}^{\infty} \frac{2^{2i-1}}{(2i-1)!} = \nonumber\\
    & =   k + 1 + \frac{n^2\cdot 4^{k-1}}{(2k)!} \cdot \sinh{2} \le \nonumber\\
    & \le k + 1 + \frac{n^2\cdot 4^{k}}{(2k)!}\text{.} \label{UpperBound1}
\end{align}
For $k=\ceil{\frac{(1+\alpha_n)\ln n}{\ln\ln n}}$, 
as $x\ln x - x$ is an increasing function for $x\ge 1$, and $\alpha_n > 0$, $k\geq 2$, we have that
\begin{align*}
\ln\brac{\frac{(2k)!}{4^{k}}} & \ge
2k \ln (2k) - 2k - k\ln 4 =
2k \ln k - 2k \ge \\
& \ge 2\cdot\frac{(1+\alpha_n)\ln n}{\ln\ln n}\cdot \ln\left(\frac{(1+\alpha_n)\ln n}{\ln\ln n}\right) - 2\cdot\frac{(1+\alpha_n)\ln n}{\ln\ln n}  \ge\\
& \ge 2\cdot\frac{(1+\alpha_n)\ln n}{\ln\ln n}\cdot  \brac{\ln\brac{\frac{\ln n}{\ln\ln n}}-1} =\\
& = 2\cdot\frac{\ln\ln n}{\ln\ln n - \ln\ln\ln n - 1}\cdot\frac{\ln n}{\ln\ln n}\cdot\brac{\ln\ln n - \ln\ln\ln n - 1} =\\
& = \ln(n^2)\text{,}
\end{align*}
and as a consequence, $\frac{(2k)!}{4^{k}} \ge n^2$, 
hence by~\eqref{UpperBound1},
\[
\EXP_{\ll}\left[\chi_{FF}(T,\ll)\right] \le k + 2
\text{.}
\]
As this holds for any forest $T$ with $n$ vertices, we conclude that
\[
\RFF_{\mathcal{F}}(n) \le \frac{(1+\alpha_n) \ln n}{2\ln\ln n}+\frac{3}{2}\text{,}
\]
which ends the proof.
\end{proof}

As $\alpha_n = \Oh{\log\log\log n\,/\,\log \log n}$ in \cref{lem:upper-bound}, we immediately get the following corollary. 
\begin{corollary}
\label{cor:upper-bound}
\[
\RFF_{\mathcal{F}}(n) \leq (\nicefrac{1}{2}+\oh{1})\cdot \ln n \, / \, \ln\ln n \,\text{.}
\]
\end{corollary}

\section{Lower bound}
\label{sec:lower-bound}

In this section, for any given $0 < \gamma < 1$, and a positive integer $k$ we construct a tree $T$, such that First-Fit uses $k$ colors to color $T$ in a random presentation order with probability at least $1-\gamma$.
Thus, $\EXP_{\ll}[\chi_{FF}(T,\ll)] \ge k(1-\gamma)$.
The number of vertices in the constructed tree is of the order $k^{dk}$ where $d$ is a constant depending on $\gamma$.
This implies that $\RFF_{\mathcal{F}}(n) = \Om{\log n \, / \, \log\log n}$.

For the purpose of analysis, we use a slightly modified, yet equivalent random model.
Namely, 
rather than choosing a permutation in a straightforward manner,
we utilize a natural two-stage process. 
We first associate with every vertex $v$ of a graph $G$ an independent random variable $X_v\sim U[0,1]$ uniformly distributed over $[0,1]$ interval.
We call $X_v$ to be the \emph{position} of a vertex~$v$.
Then, we order the vertices according to their positions, i.e., so that $u\ll v$ if and only if $X_u\leq X_v$.
Note that such an order is uniquely determined with probability $1$.
Moreover, by the symmetry of the process, each resulting vertex permutation is equiprobable.

\begin{lemma}
\label{lem:lower-bound}
For the class $\mathcal{F}$ of forests, we have:
\[
\RFF_{\mathcal{F}}(n) \ge (\nicefrac{1}{2}-\oh{1})\cdot \ln n \, / \, \ln\ln n \,\text{.}
\]
\end{lemma}

\begin{proof}

Fix any $0 < \gamma < 1$, and any integer $k \geq 3$.
Let $c = \nicefrac{10}{\gamma^2}$ and $r = \ceil{c k \ln k}$.
Set $\epsi_i = \nicefrac{i\gamma}{k}$ for $i=1,2,\ldots,k$.
We recursively define rooted trees $T^r_1, T^r_2, \ldots, T^r_k$ as follows.
The tree $T^r_1$ is a single vertex, and for $i=1,2,\ldots,k-1$, the
tree $T^r_{i+1}$ is constructed of $r$ copies of each of the trees $T^r_j$ with $j=1,2\ldots,i$ by joining their roots to a single additional vertex -- the root of $T^r_{i+1}$ (hence the root of $T^r_{i+1}$ has degree $ri$).
See \cref{fig:tree}.

Consider First-Fit coloring of $T^r_i$ in a random presentation order, for any fixed $i \in \set{1,2,\ldots,k}$.
We assume that the presentation order is given by the positions $X_v \in [0,1]$ drawn uniformly at random for every vertex $v$ of $T^r_i$.
By the construction of $T^r_i$, the longest simple path ending at the root of $T^r_i$ includes at most $i$ vertices.
From \cref{obs:onepath} we immediately obtain the following claim.
\begin{claim}\label{clm:rooti}
The color assigned by First-Fit to the root of $T^r_i$ when coloring $T^r_i$ in any order, does not exceed $i$.
\end{claim}
We define $B_i$ to be the random event that First-Fit assigns color smaller then $i$ to the root vertex of $T^r_i$ when coloring $T^r_i$ in a random presentation order.
\begin{claim}\label{clm:success}
    For every $i=1,2,\ldots,k$ we have
    \[
    \mathbb{P}(B_i) \le \epsi_i = \frac{i\gamma}{k}\text{.}
    \]
\end{claim}
\begin{claimproof}
We prove the claim by induction on $i$.
For $i=1$, First-Fit assigns color $1$ to the only vertex of $T^r_1$ and hence $\mathbb{P}(B_1)=0$.
Now, for the induction step, we fix any $1 \le i \le k-1$, and assume that 
\begin{align}\label{BjAssumptions}
   \mathbb{P}(B_j) \le \epsi_j 
\end{align} 
holds for every $j=1,2,\ldots,i$.
We shall prove that $\mathbb{P}(B_{i+1}) \le \epsi_{i+1}$.
For that we focus on the coloring of $T=T_{i+1}^r$.
Denote the root of $T$ by $v$ and let $v_j^{(q)}$ with $j=1,\ldots,i$, $q=1,\ldots,r$ be the neighbors of $v$ in $T$ where each $v_j^{(q)}$ is the root of one of the $r$ copies of $T^r_j$ attached to $v$ -- we denote this copy as $T^{(q)}_j$.
See \cref{fig:tree}.
We denote the (random) color First-Fit assigns to any vertex $w$ by $c(w)$.

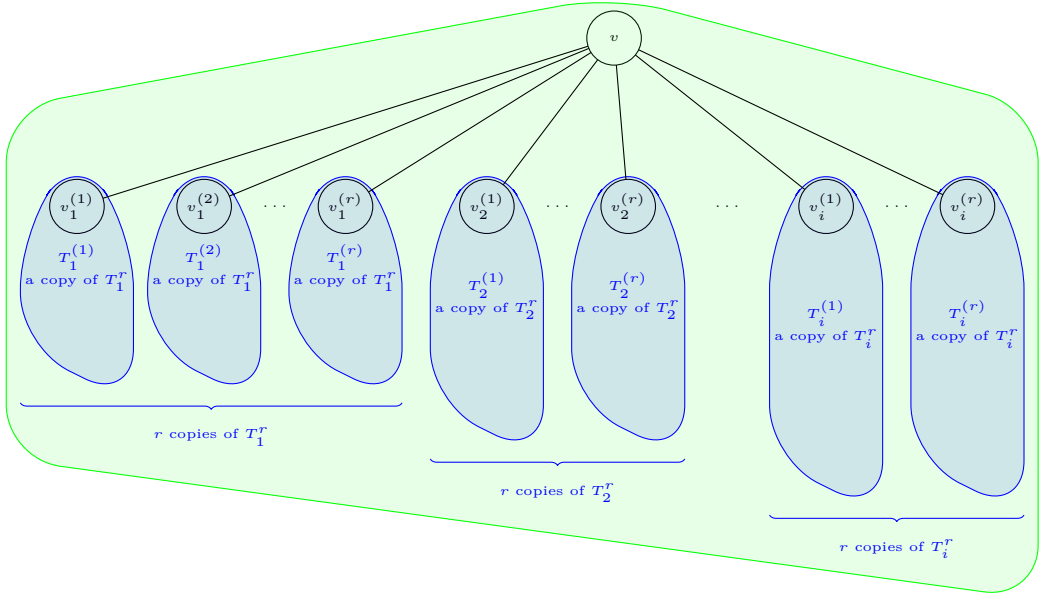
\begin{figure}
\begin{center}
\begin{tikzpicture}
\begin{tiny}
\begin{scope}[scale=\textwidth/1cm,>=latex,line join=bevel]
\begin{scope}[scale=1/18,xshift=0,yshift=0]
\tikzset{
vertex/.style = {circle, draw, minimum size=3em, inner sep=0em, outer sep=0em},
}
\node [vertex] (root) at (  0em,  0em) {$v$};
\draw [rounded corners=3em, fill=green, fill opacity=0.1, draw=green] (0em,3em) -- (30em, -5em) -- (30em,-40em) -- (-43em,-30em) -- (-43em,-5em) -- cycle;

\begin{scope}[yshift=-12em]
\begin{scope}[xshift=-24em]
\begin{scope}[xshift=-14em]
\node [vertex] (v11) at (0em, 0em) {$v_{1}^{(1)}$};
\draw [rounded corners=3em, fill=blue, fill opacity=0.1, draw=blue] (0em,4em) -- (4em, -2em) -- (4em,-14em) -- (-4em,-10em) -- (-4em,-2em) -- cycle;
\node [fit={(0em,4em) (4em, -2em) (4em,-14em) (-4em,-10em) (-4em,-2em)}, blue] {$T_{1}^{(1)}$\\a copy of $T_1^r$};
\draw (v11) -- (root);
\end{scope}
\begin{scope}[xshift=-5em]
\node [vertex] (v12) at (0em, 0em) {$v_{1}^{(2)}$};
\draw [rounded corners=3em, fill=blue, fill opacity=0.1, draw=blue] (0em,4em) -- (4em, -2em) -- (4em,-14em) -- (-4em,-10em) -- (-4em,-2em) -- cycle;
\node [fit={(0em,4em) (4em, -2em) (4em,-14em) (-4em,-10em) (-4em,-2em)}, blue] {$T_{1}^{(2)}$\\a copy of $T_1^r$};
\draw (v12) -- (root);
\end{scope}
\node [circle, minimum size=3em, inner sep=0em, outer sep=0em] (v1d) at (0em, 0em) {$\ldots$};
\begin{scope}[xshift=5em]
\node [vertex] (v1r) at (0em, 0em) {$v_{1}^{(r)}$};
\draw [rounded corners=3em, fill=blue, fill opacity=0.1, draw=blue] (0em,4em) -- (4em, -2em) -- (4em,-14em) -- (-4em,-10em) -- (-4em,-2em) -- cycle;
\node [fit={(0em,4em) (4em, -2em) (4em,-14em) (-4em,-10em) (-4em,-2em)}, blue] {$T_{1}^{(r)}$\\a copy of $T_1^r$};
\draw (v1r) -- (root);
\end{scope}
\draw[decoration={brace,mirror},decorate, blue] (-18em,-14em) -- node[below=1em] {$r$ copies of $T_1^r$} (9em,-14em);
\end{scope}

\begin{scope}[xshift=-4em]
\begin{scope}[xshift=-5em]
\node [vertex] (v21) at (0em, 0em) {$v_{2}^{(1)}$};
\draw [rounded corners=3em, fill=blue, fill opacity=0.1, draw=blue] (0em,4em) -- (4em, -2em) -- (4em,-18em) -- (-4em,-14em) -- (-4em,-2em) -- cycle;
\node [fit={(0em,4em) (4em, -2em) (4em,-18em) (-4em,-14em) (-4em,-2em)}, blue] {$T_{2}^{(1)}$\\a copy of $T_2^r$};
\draw (v21) -- (root);
\end{scope}
\node [circle, minimum size=3em, inner sep=0em, outer sep=0em] (v1d) at (0em, 0em) {$\ldots$};
\begin{scope}[xshift=5em]
\node [vertex] (v2r) at (0em, 0em) {$v_{2}^{(r)}$};
\draw [rounded corners=3em, fill=blue, fill opacity=0.1, draw=blue] (0em,4em) -- (4em, -2em) -- (4em,-18em) -- (-4em,-14em) -- (-4em,-2em) -- cycle;
\node [fit={(0em,4em) (4em, -2em) (4em,-18em) (-4em,-14em) (-4em,-2em)}, blue] {$T_{2}^{(r)}$\\a copy of $T_2^r$};
\draw (v2r) -- (root);
\end{scope}
\draw[decoration={brace,mirror},decorate, blue] (-9em,-18em) -- node[below=1em] {$r$ copies of $T_2^r$} (9em,-18em);
\end{scope}

\node [circle, minimum size=3em, inner sep=0em, outer sep=0em] (v1d) at (8em, 0em) {$\ldots$};

\begin{scope}[xshift=20em]
\begin{scope}[xshift=-5em]
\node [vertex] (vi1) at (0em, 0em) {$v_{i}^{(1)}$};
\draw [rounded corners=3em, fill=blue, fill opacity=0.1, draw=blue] (0em,4em) -- (4em, -2em) -- (4em,-22em) -- (-4em,-18em) -- (-4em,-2em) -- cycle;
\node [fit={(0em,4em) (4em, -2em) (4em,-22em) (-4em,-18em) (-4em,-2em)}, blue] {$T_{i}^{(1)}$\\a copy of $T_i^r$};
\draw (vi1) -- (root);
\end{scope}
\node [circle, minimum size=3em, inner sep=0em, outer sep=0em] (v1d) at (0em, 0em) {$\ldots$};
\begin{scope}[xshift=5em]
\node [vertex] (vir) at (0em, 0em) {$v_{i}^{(r)}$};
\draw [rounded corners=3em, fill=blue, fill opacity=0.1, draw=blue] (0em,4em) -- (4em, -2em) -- (4em,-22em) -- (-4em,-18em) -- (-4em,-2em) -- cycle;
\node [fit={(0em,4em) (4em, -2em) (4em,-22em) (-4em,-18em) (-4em,-2em)}, blue] {$T_{i}^{(r)}$\\a copy of $T_i^r$};
\draw (vir) -- (root);
\end{scope}
\draw[decoration={brace,mirror},decorate, blue] (-9em,-22em) -- node[below=1em] {$r$ copies of $T_i^r$} (9em,-22em);
\end{scope}

\end{scope}

\end{scope}
\end{scope}
\end{tiny}
\end{tikzpicture}
\end{center}
\caption{Recursive construction of the tree $T_{i+1}^r$.}
\label{fig:tree}
\end{figure}

Suppose the position of $v$ is fixed and equals $x$.
Note that if for some $1 \le j \le i$ there is some $q$ such that $c(v_j^{(q)})\geq j$ and the position of $v_j^{(q)}$ is smaller than $x$, then $c(v) \neq j$.
Indeed, due to \cref{obs:onepath} and the fact that $v_j^{(q)}$ is positioned before $v$ we have that $c(v_j^{(q)}) \le j$ and it is assigned with color $j$ by the assumption $c(v_j^{(q)})\geq j$.
Therefore, if there is such a $q$ for every $1 \le j \le i$, then we get $c(v) \ge i+1$.
This allows for the following inequality,

\begin{align}
    \mathbb{P}\left(\overline{B_{i+1}}~|~X_v=x\right)
    &\ge  \mathbb{P}\left(\forall j\leq i ~\exists q\leq r: c(v_j^{(q)}) \ge j \wedge X_{v_j^{(q)}} < x  ~|~X_v=x\right) =\nonumber\\
    &=  \prod_{j=1}^{i} \mathbb{P}\left(\exists q\leq r: c(v_j^{(q)}) \ge j \wedge X_{v_j^{(q)}} < x  ~|~X_v=x\right) =\nonumber\\
    &= \prod_{j=1}^{i} \left(1- \prod_{q=1}^{r}\mathbb{P}\left(c(v_j^{(q)}) < j \vee X_{v_j^{(q)}} > x  ~|~X_v=x\right)\right) \label{ProbOverlineBi}
\end{align}
where the two last equalities above follow by the independence of the corresponding events for a fixed value of $x$.
Obviously $\mathbb{P}(X_{v_j^{(q)}} > x  ~|~X_v=x) = 1-x$, as positions of the vertices are independent.
Further, for any assignment $X$ of positions to all vertices of $T$, one can consider First-Fit coloring of $T_j^{(q)}$ in order given by the restriction of $X$ to the vertices of $T_j^{(q)}$.
Color assigned to the root of $T_j^{(q)}$ in this restricted coloring is not greater then the color assigned to this vertex in the coloring of $T$.
Thus, $\mathbb{P}(c(v_j^{(q)}) < j ~|~X_v=x)\leq \mathbb{P}(B_j)$, and by~\eqref{BjAssumptions}, for every $j\le i$, we have:
\begin{align}\label{ProbSingle}
    \mathbb{P}\left(c(v_j^{(q)}) < j \vee X_{v_j^{(q)}} > x  ~|~X_v=x\right) \leq \min\set{1,1-x+\varepsilon_j} \text{.}
\end{align}
By~\eqref{ProbOverlineBi} and~\eqref{ProbSingle}, we thus obtain that
\begin{align}
    \mathbb{P}\left(B_{i+1}\right) &= 1 - \mathbb{P}\left(\overline{B_{i+1}}\right) 
    = 1 - \int_{0}^{1} \mathbb{P}\left(\overline{B_{i+1}}~|~X_v=x\right) dx  \le \nonumber\\
&\le 1-\int_{0}^{1} \prod_{j=1}^{i} \max\left\{0,1-(1-x+\epsi_j)^r)\right\} dx \le \nonumber\\
&\le
1-\int_{\epsi_{i}}^{1} (1-(1-x+\epsi_i)^r)^{i} dx\text{,} \label{PBi+1-1st_ineq}
\end{align}
where the last inequality follows by the fact that $\epsi_j\le\epsi_i$ for $j\le i$.
Substituting $y=1-x+\epsi_i$ in~\eqref{PBi+1-1st_ineq} we further obtain:
\begin{align}
\mathbb{P}(B_{i+1}) \le
1-\int_{\epsi_{i}}^{1} (1-y^r)^i dy\text{.} \label{PBi+1-2nd_ineq}
\end{align}
Let $f(y) = (1-y^r)^i$.
Observe that $f(0)=1$, $f(1)=0$, and $f$ is strictly decreasing in $[0,1]$.
Thus, we obtain that $\int_0^{\epsi_i} f(y) dy \leq \epsi_i$, and by~\eqref{PBi+1-2nd_ineq},
\begin{align}
P(B_{i+1}) &\le 
\epsi_i + 1 - \int_{0}^{1} f(y) dy \le
\epsi_i + 1 - \left(1-\frac{\gamma}{2k}\right)\cdot f\left(1-\frac{\gamma}{2k}\right) \le\nonumber\\*
&\le
\epsi_i + \frac{\gamma}{2k} + 1-f\left(1-\frac{\gamma}{2k}\right)\text{.}
\label{PBi+1-3rd_ineq}
\end{align}
In order to prove that $P(B_{i+1}) \le \epsi_{i+1} = \epsi_i + \frac{\gamma}{k}$ it thus remains to show that $f(1-\frac{\gamma}{2k}) \geq 1-\frac{\gamma}{2k}$.
Note that
\[
f\left(1-\frac{\gamma}{2k}\right) = \left(1 - \left(1-\frac{\gamma}{2k}\right)^r\right)^i 
\ge \left(1 - \left(1-\frac{\gamma}{2k}\right)^{ck\ln k}\right)^k 
= \left(1-\left(1-\frac{\gamma}{2k}\right)^{\frac{2k}{\gamma}\frac{\gamma}{2}c\ln k}\right)^k\text{.}
\]
Now, for $\alpha = \frac{2k}{\gamma} \geq 1$ we have that $(1-\frac{1}{\alpha})^\alpha < \frac{1}{e}$.
Thus,
\[
f\left(1-\frac{\gamma}{2k}\right) \geq \left(1-\frac{1}{e^{\frac{\gamma c\ln k}{2}}}\right)^k 
= \left(1-\frac{1}{k^{\frac{\gamma c}{2}}}\right)^k\text{.}
\]
As for $0 < \beta = \frac{1}{k^{\frac{\gamma c}{2}}} < 1$ we have that $1-\beta > e^{-\frac{\beta}{1-\beta}}$, we thus further obtain that
\[
f\left(1-\frac{\gamma}{2k}\right) \geq e^{-\frac{\beta k}{1-\beta}} = e^{-\frac{k}{k^\frac{\gamma c}{2}-1}} \geq 1-\frac{k}{k^\frac{\gamma c}{2}-1}\text{.}
\]
Now, for $c=\frac{10}{\gamma^2}$, and using $k \ge 3$ we get
\[
f\left(1-\frac{\gamma}{2k}\right) \geq 1-\frac{k}{k^{\frac{5}{\gamma}}-1} \geq 1-\frac{1}{k^\frac{3}{\gamma}}\text{.}
\]
Observe that $k^{\frac{3}{\gamma}} > \frac{2k}{\gamma}$ for $\gamma \in [0,1]$, and $k \ge 3$.
This finally gives that $f(1-\frac{\gamma}{2k}) \geq 1-\frac{\gamma}{2k}$, and consequently, by~\eqref{PBi+1-3rd_ineq}, $P(B_{i+1}) \le \epsi_{i+1}$, which ends the proof of \cref{clm:success}.
\end{claimproof}
By Claims~\ref{clm:rooti} and~\ref{clm:success}, First-Fit assigns color $k$ to the root of $T^r_k$ with probability at least  $1-\varepsilon_k$. Using Claim~\ref{clm:rooti} (or \cref{obs:onepath}), one can thus easily deduce the following claim.
\begin{claim}\label{cl:Trkexactlyk}
    $\mathbb{P}(\chi_{FF}(T^r_k,\ll)=k)\geq 1-\gamma$.
\end{claim}

\begin{claim}\label{cl:OrderofTrk}
    $T^r_k$ has exactly $(r+1)^{k-1}$ vertices.
\end{claim}
\begin{claimproof}\renewcommand\qedsymbol{\textcolor{lipicsGray}{\ensuremath{\vartriangleleft}}}
    We prove the claim by induction on $k$.
For $k=1$ the claim trivially holds.
Let us assume that $k\geq 1$ and that the claim holds for all the trees $T_1^r,\ldots,T_k^r$.
By the construction of $T^r_{k+1}$,
\begin{align*}
  |T^r_{k+1}| &= 1 + r \cdot \sum_{i=1}^{k} |T^r_i| 
  =  1 + r \cdot \sum_{i=1}^{k} (r+1)^{i-1}
  = 1 + r \cdot \frac{1-(r+1)^{k}}{1-(r+1)}
  = (r+1)^{k}\text{.}\qedhere
\end{align*}\end{claimproof}

For any value of $k$, we set $\gamma = \frac{1}{\ln k}$, and for this choice of $\gamma$ we get $c=10\ln^2k$ and $r=\ceil{10k\ln^3k}$.
Let us denote by $n_k$ the number of vertices in $T^r_k$.
Then the following hold for $k$ large enough. Firstly, by Claim~\ref{cl:OrderofTrk}, we have $\ln n_k \ge k$.
Consequently, again by Claim~\ref{cl:OrderofTrk},
\begin{align}\label{logn-upper_bound}
    \ln n_k &\le (k-1)\ln\left(2+10k\ln^3k\right) 
    \le k \left(\ln k + 4\ln\ln k\right) \le\nonumber\\
    &\le k \left(\ln\ln n_k + 4\ln\ln\ln n_k\right)\text{.}
\end{align}
By Claim~\ref{cl:Trkexactlyk} and~\eqref{logn-upper_bound}, we thus finally obtain that 
\begin{align}
 \EXP_{\ll}\left[\chi_{FF}(T^r_k,\ll)\right]  
 &\ge k \cdot \left(1-\frac{1}{\ln k}\right) \label{final_lower1} \ge\\
 &\ge \frac{\ln n}{\ln\ln n + 4\ln\ln\ln n} \cdot
 \left(1-\frac{1}{\ln\ln n - 2\ln\ln\ln n}\right) 
 = g(n)\text{.}\label{final_lower2}
\end{align}
Note that $g(n)$ is an increasing function for large enough $n$.
Note also that $n_k$ is also an increasing function of $k$.
Consider any (large enough) $n$ such that $n_k\le n \le n_{k+1}$ for some $k$.
Then, since $\RFF_{\mathcal{F}}(n)$ is a nondecreasing function of $n$, by~\eqref{final_lower1} and~\eqref{final_lower2},
\begin{align*}
 \RFF_{\mathcal{F}}(n) &\ge \RFF_{\mathcal{F}}(n_1)  
 \ge \frac{1}{2}\cdot\EXP_{\ll}\left[\chi_{FF}(T^r_k,\ll)\right] 
 \ge \frac{1}{2}\cdot k \cdot \left(1-\frac{1}{\ln k}\right) \ge\\
 &\ge \frac{1}{2}\cdot\brac{(k+1) \cdot \left(1-\frac{1}{\ln (k+1)}\right) - 1}
 \ge \frac{1}{2}\cdot(g(n_{k+1}) - 1) \ge\\
 &\ge \frac{1}{2}\cdot(g(n) - 1)
 = \left(\frac{1}{2}-\oh{1}\right)\frac{\ln n}{\ln\ln n}\text{,}
\end{align*}
which finishes the proof of Lemma~\ref{lem:lower-bound}.
\end{proof}

\section{Final comments}
Finally, \cref{thm:main} is obtained by combining the matching bounds of Corollary~\ref{cor:upper-bound} and Lemma~\ref{lem:lower-bound}.
This concludes our investigation of the efficiency of First-Fit in the random arrival model on the class of forests.
First-Fit is efficient on this class even in the adversarial model.
Still, the randomization of the presentation order allows for some increase in performance.
This raises the hope that First-Fit is more effective in the average case than it is in the worst case also on some other graph classes.

The systematic analysis of First-Fit in the random arrival model on other graph classes should continue for the class $\mathcal{B}$ of bipartite graphs.
There, First-Fit is known to be extremely inefficient in the adversarial model with $\AFF_{\mathcal{B}} = \OT{n}$.
However, there is another simple algorithm~\cite{LovaszST89}, a clever modification of First-Fit, that uses $\Oh{\log n}$ colors to color any bipartite graph with $n$ vertices.
A natural extension of our research would be to assess the First-Fit algorithm in the random arrival model on bipartite graphs.
We anticipate that the outcome will not deviate significantly from the results obtained for forests.

\begin{conjecture}
First-Fit in the random arrival model uses $\poly (\log n)$ colors in expectation when coloring any bipartite graph with $n$ vertices.
\end{conjecture}

Addressing the aforementioned conjecture represents a significant cognitive pursuit.
However, a truly remarkable feat would be to establish some nontrivial bounds for 3-colorable graphs.
We want to thank Anna Zych-Pawlewicz for inspiring the work on this project.

\bibliography{paper}
\end{document}